\documentclass{iacrcc}
\usepackage{graphicx} 
\usepackage{appendix}
\usepackage{amsmath}
\usepackage{amsthm}
\usepackage{amsfonts}
\usepackage{amssymb}
\usepackage{amsmath,mathtools}
\usepackage{paralist}
\usepackage{xcolor}
\usepackage{setspace} 
\usepackage{todonotes}
\usepackage{makecell}
\usepackage{siunitx}
\usepackage{multirow}

\def\q{{\mathfrak q}}
\def\fq{{\mathbb F}_{q}}
\def\f2{{\mathbb F}_{2}}

\def\fF{{\mathbb F}_q}

\newcommand{\F}[1]{{\mathbb F}_q^{#1}}

\def\sd{{\mathrm{sd}}}

\title{A new multivariate primitive from CCZ equivalence}
\addauthor[orcid = {0000-0002-6817-3421},
inst = {1},
email = {marco.calderini@unitn.it},
surname = {Calderini}]{Marco Calderini}
\addauthor[orcid = {0000-0001-5227-807X},
inst = {2},
email = {alessio.caminata@unige.it},
surname = {Caminata}]{Alessio Caminata}

\addauthor[orcid = {0000-0001-6381-4712},
inst = {1,2},
email = {irene.villa@unitn.it},
surname = {Villa}]{Irene Villa}
\addaffiliation{Dipartimento di Matematica,  Universit\`a degli studi di Trento\\ via Sommarive 14, 38123 Povo, Trento, Italy}

\addaffiliation{Dipartimento di Matematica, Universit\`a di Genova\\ via Dodecaneso 35, 16146, Genova, Italy}


\begin{document}

\maketitle

\begin{abstract}
    Multivariate Cryptography is one of the candidates for Post-quantum Cryptography. 
    Multivariate schemes are usually constructed by applying two secret affine invertible transformations $\mathcal S,\mathcal T$ to a set of multivariate polynomials $\mathcal{F}$ (often quadratic). The polynomials  $\mathcal{F}$ possess a trapdoor that allows the legitimate user to find a solution of the corresponding system, while the public polynomials $\mathcal G=\mathcal S\circ\mathcal F\circ\mathcal T$ look like random polynomials. The polynomials $\mathcal G$ and $\mathcal F$ are said to be affine equivalent.
    In this article, we present a more general way of constructing a multivariate scheme by considering the CCZ equivalence, which has been introduced and studied in the context of vectorial Boolean functions. 
\end{abstract}
\section{Introduction}

In the last few decades, many new proposals for public-key cryptosystems have been presented to the scientific community. With the advent of \emph{Post-quantum Cryptography} \cite{BernsteinPQ} following the development of Shor's algorithm, many cryptographers have focused on finding quantum-resistant public-key systems. \emph{Multivariate public-key Cryptography} is one of the main families of post-quantum cryptosystems. These systems base their security on the difficulty of solving a set of randomly chosen nonlinear multivariate polynomials over a finite field. So far, there is no evidence that quantum computers can solve such sets of multivariate polynomials efficiently.

A multivariate public-key cryptosystem involves a public key comprising multivariate polynomials  $f^{(1)},\ldots,f^{(m)}$ in $\fF[z_1,\ldots,z_n]$, where $\fF$ is a finite field with $q$ elements. In order to keep the public key size not too large, usually, quadratic polynomials are considered.
The secret key is some information (about the construction of the polynomials $f^{(i)}$) which allows to solve the system $f^{(1)}=w_1,\ldots,f^{(m)}=w_m$ efficiently for some $w_1,\ldots,w_m\in\fF$.
To encrypt a message $(z_1',\ldots,z_n')\in\F n$, the sender computes $w_i=f^{(i)}(z_1',\ldots,z_n')$, for $i=1,\ldots,m$, and sends $(w_1,\ldots,w_m)$ to the receiver. With the secret key, the receiver can solve the system and recover the original message.

One of the main methods to achieve the previous scheme is the {\em Bipolar Construction}.
From a system $\mathcal F$ of $m$ equations in $n$ variables  (possibly quadratic) relatively easy to invert, the secret key is composed of $\mathcal F$ and two randomly chosen affine bijections $\mathcal S$ and $\mathcal T$; the public key is the system  $\mathcal G=\mathcal S\circ\mathcal F\circ\mathcal T$.
The obtained system $\mathcal G$ (still quadratic) is now assumed to be not easy to invert, since it should be hard to distinguish from a random system.
Notice that, in this case, finding a preimage of $(w_1,\ldots,w_m)$ reduces to finding a preimage for $\mathcal F$ of $\mathcal S^{-1}(w_1,\ldots,w_m)$ and then apply $\mathcal T^{-1}$.
Hence, the core idea of the Bipolar Construction is to hide the structure of the {\em central map} $\mathcal F$ by applying two random affine bijections to the input and to the output of $\mathcal F$. 
This corresponds to randomly taking a system in the affine-equivalence class of $\mathcal F$. 

The Bipolar Construction method has found extensive application in numerous significant multivariate schemes, including Matsumoto--Imai \cite{MI}, HFE \cite{HFE}, Oil and Vinegar \cite{Pat97}, and Rainbow \cite{rainbow}.
Unfortunately, the affine equivalence keeps many properties and structures of a system of equations. This might allow an attacker to use them to break the system. Consequently, public polynomials 
$\mathcal{G}$ often fail to exhibit true randomness, undermining the scheme's robustness.

In order to better explain this phenomenon, we briefly recall the example of the  Matsumoto--Imai (MI) Cryptosystem \cite{MI}.
In this scheme, the authors consider $q$ a power of two and $m=n$.
The central system $\mathcal F$ is seen as a quadratic function  $F:\mathbb E\rightarrow\mathbb E$, $F(z)=z^{q^j+1}$ where $\mathbb E=\mathbb{F}_{q^n}$ and $j$ is chosen such that  $\gcd(q^n-1,q^j+1)=1$. 
Recall that the degree of a monomial $z^k\in\mathbb E[z]$, $0\le k\le q^n-1$, is the 2-weight of $k$.
Then, by applying the standard isomorphism $\phi:\F n\rightarrow\mathbb E$, the map is
 expanded as a system of $n$ equations over $\fF$, $\mathcal F =\phi\circ F\circ\phi^{-1}:\F n\rightarrow\F n$.
 The quadratic function $F(z)=z^{q^j+1}$ is a bijection and it is easy to invert.
Clearly, after applying the affine bijections $\mathcal S$ and $\mathcal T$ to the input and to the output of $\mathcal F$, the system will not have a so-simple structure easy to invert.
However, the function $F$ presents a bilinear relation between its input and its output. Indeed, setting $w=F(z)$, we have the relation $w^{q^j}z=z^{q^{2j}}w$,
in which both variables $z$ and $w$ appear only with ``linear exponents''.
The affine transformations $\mathcal{S}$ and $\mathcal{T}$ preserve the presence of such linear relations between input and output, making systems susceptible to linearization attacks as first demonstrated in \cite{Pat95}. Modifications to the MI cryptosystem have been attempted to mitigate this vulnerability, as seen in \cite{Sflash,PMI}. However, even with these alterations, attacks have been developed that can reduce the modified MI system back to its original form, thus rendering it vulnerable to linearization attacks, as detailed in \cite{Sflashattack, FGS05,oygarden24}.
This exemplifies just one instance of a property that persists through an affine equivalence relation. 
Cryptographers have identified many such properties, which have been exploited in successful attacks on multivariate schemes, including the Min Rank attack documented in \cite{Bettale20131,Caminata2021163, GoubinCourtois, KS98, Vates2017272}.

In this paper, we propose to use a more general notion of equivalence relation between polynomial systems to obtain a public-key function $\mathcal{G}$ which does not inherit the ``simple'' structures of the secret function $\mathcal F$.
Specifically, our investigation focuses on the CCZ equivalence transformation which has been introduced by Carlet, Charpin, and Zinoviev \cite{CCZ}.
Given two functions $F,G:\F n\rightarrow \F m$ we say that they are  {\em CCZ equivalent}  if there exists an affine bijection $\mathcal A$ of $\F{n+m}$ such that $\mathcal G_G=\mathcal A(\mathcal G_F)$, where $\mathcal G_F$ and $\mathcal G_G$ are the graphs of $F$ and $G$ respectively.
This equivalence relation has been mainly studied in the context of cryptographic Boolean functions, since it keeps unchanged the value of the differential uniformity and the nonlinearity, two important properties to study when a function is used as a component of a block cipher. Budaghyan, Carlet, and Pott proved that CCZ equivalence is strictly more general than the affine equivalence by exhibiting functions which are
CCZ equivalent to $F(z) = z^3$ over $\mathbb{F}_{2^n}$, but not affine equivalent, and also not extended affine (EA) equivalent \cite{BCP2006}, where EA equivalence extends affine equivalence by adding any affine transformation.
Other works focused on further studying the relation between CCZ and EA equivalences, see for example \cite{CanPer,BCV20}.
 In this work, we propose constructing multivariate schemes in which the secret and public maps are CCZ equivalent but not necessarily affine or extended affine equivalent. While CCZ equivalence preserves the cryptographic properties mentioned above, these may not be relevant in this context. Instead, we focus on the fact that CCZ equivalence is more general than affine equivalence and does not preserve other key properties such as algebraic degree and bijectivity. The central idea of this paper is to investigate whether the relations usually used to break multivariate systems---typically preserved under affine transformations---are instead disrupted by CCZ transformations.

One of the main challenges in transitioning from affine equivalence to CCZ equivalence is that selecting a random element $G$ from the CCZ class of a given polynomial map $F$ is not straightforward. Not every affine bijection of $\F{n+m}$  maps the graph of one function into the graph of another, and the admissible affine bijections depend on the chosen function $F$. To address this challenge, we leverage a result by Canteaut and Perrin \cite{CanPer}, demonstrating that any two CCZ equivalent functions can be connected by applying two extended affine transformations and another map called a $t$-twist (see Definition~\ref{def:ttwist}). In Section~\ref{sec:proposal}, we provide a high-level explanation of how this strategy can be utilized to construct an encryption or a signature scheme.
This construction is quite general and applicable to any secret map 
$F$ that admits a $t$-twist. To provide a concrete example, we propose selecting a quadratic function $F$ derived from Oil and Vinegar (OV) polynomials \cite{Pat95}.
OV polynomials divide variables into oil and vinegar sets, with no quadratic terms in the oil variables, allowing for linear decryption/signing by assigning random values to the vinegar variables. The name comes from the fact that the variables do not truly mix, like oil and vinegar in the salad dressing.
We name our scheme \texttt{Pesto}, as the CCZ transformation ensures that the variables fully mix, resembling the mixing of ingredients in Pesto Sauce using mortar and pestle. Notably, while the secret polynomials are quadratic, the public polynomials have degree four.

We provide commentary on the proposal, specifically regarding the dimensions of the keys and the computational costs involved, in Section~\ref{sec:computationalrmks}. Discussion on potential vulnerabilities and attacks is reserved for the final section of the paper.
Given that the public polynomials have degree four, many typical attacks against multivariate schemes, which target degree two, are not immediately applicable. However, partial recovery of the affine transformation may still be possible, albeit at significant computational expense.

Although our primary objective is to connect the research areas of Boolean functions and Multivariate Cryptography rather than presenting a fully developed proposal, we also explore specific parameter choices for the system. These parameters are not intended for immediate implementation but rather serve as an invitation for further cryptanalysis of the scheme.
We believe this area offers significant opportunities for further research and development, and we warmly encourage additional contributions.

\subsubsection*{Structure of the paper}
This work is organized as follows.
In Section~\ref{sec:prelim} we recall the definition and some basic facts about the CCZ equivalence.
In Section~\ref{sec:twisting}, we present the \emph{twisting} and explain how it can be used to produce a random CCZ transformation.
In Section~\ref{sec:proposal} we present a proposal for a multivariate cryptographic scheme obtained by hiding the central map with a CCZ transformation. We present it at a high level of generality in Subsection~\ref{subsec:genproposal}, and then with more specifics on the functions in Subsection~\ref{subsec:ourproposal}, i.e., the \texttt{Pesto} scheme.
Section~\ref{sec:computationalrmks} presents some comments on the form of the constructed maps, an analysis on the sizes of the keys and on the computational cost of applying the proposed procedure.
In the last section, we present a preliminary analysis on the security of the scheme and propose a possible choice of secure parameters.

\subsubsection*{Acknowledgments}
A.~Caminata is supported by the Italian PRIN2022 grant P2022J4HRR ``Mathematical Primitives for Post Quantum Digital Signatures'', by the MUR Excellence Department Project awarded to Dipartimento di Matematica, Università di Genova, CUP D33C23001110001, and by the European Union within the program NextGenerationEU. Additionally, part of the work was done while Caminata was visiting the Institute of Mathematics of the University of Barcelona (IMUB). He gratefully appreciates their hospitality during his visit.\\
The research of M.~Calderini and I.~Villa was partially supported by the Italian Ministry of University and Research with the project PRIN 2022RFAZCJ.\\
The three authors are members  of the INdAM Research Group GNSAGA.\\
We sincerely thank the anonymous referees for their insightful comments and for encouraging us to propose specific parameters for the scheme.


\section{Preliminaries}\label{sec:prelim}
In this section, we recall the definition of CCZ equivalence, introduced in \cite{CCZ}, together with some preliminary results and notations.

\subsection{CCZ equivalence}
Let $n,m$ be positive integers, $q$ a prime power and $\fF$ a finite field with $q$ elements. We consider a function $F:\F n\rightarrow \F m$.
Notice that $F$ can be seen as
$F=(f_1,\ldots,f_m)$
with $f_i:\F n\rightarrow\fF$ for $1\le i\le m$.
We call  \emph{the coordinate}, or the \emph{the $i$-th coordinate}, of $F$ the function $f_i$.
For $\lambda=(\lambda_1,\ldots,\lambda_m)\in\F m$ we call the \emph{$\lambda$-component} of $F$ the function $F_\lambda=\lambda\cdot F=\lambda_1f_1+\cdots+\lambda_mf_m$.
To represent $F$ we can use the {\em algebraic normal form} (ANF), that is, we can represent the function as a multivariate polynomial over $\F m$:
$$F(z)=F(z_1,\ldots,z_n)=
\sum_{u\in\mathbb{N}^n}a_uz_1^{u_1}\cdots z_n^{u_n}, \qquad \mbox{ with }a_u\in\F m.$$
Moreover, in order to have a unique representative for $F$ we adopt the standard convention that $u_1,\dots,u_n<q$.
We say that $z_1^{u_1}\cdots z_n^{u_n}$ is a \emph{term} of $F$ if $a_u\ne0$.
The {\em algebraic degree} of $F$ is $\deg(F)=\max\{\sum_{i=1}^nu_i\, :\, u\in\mathbb{N}^n;\, a_u\ne0\}$.
We call $F$ {\em linear} if $\deg(F)=1$ and $F(0)=0$, {\em affine} if $\deg(F)= 1$, {\em quadratic} if $\deg(F)=2$.
When $m=n$, we say that $F:\F n\rightarrow\F n$ is  a {\em bijection} or that  it  is {\em invertible} if $F$ induces a permutation over $\F n$ , i.e., \ if $\{F(v) : v \in\F n\}=\F n$.\\

\begin{definition}\label{def:eq}
Let $F,G:\F n\rightarrow \F m$ be two functions.
    \begin{itemize}
        \item $F$ is {\em affine equivalent} to $G$ if there are two affine bijections $A_1,A_2$ of $\F m$ and $\F n$ respectively such that
        $G=A_1\circ F\circ A_2$.
        \item $F$ is {\em EA equivalent} (extended affine) to $G$ if there are two affine bijections $A_1,A_2$ of $\F m$ and $\F n$ respectively and an affine transformation $A:\F n\rightarrow\F m$ such that
        $G=A_1\circ F\circ A_2+A$.
        \item $F$ is {\em CCZ equivalent} to $G$ if there exists an affine bijection $\mathcal A$ of $\F{n+m}$ such that
        $\mathcal G_G=\mathcal A(\mathcal G_F)$, where $\mathcal G_F=\{(v,F(v)) : v\in\F n\}\subseteq\F n\times\F m$ is the graph of $F$, and  $\mathcal G_G$ is the graph of $G$.
    \end{itemize}
\end{definition}
Clearly, affine equivalence is a particular case of EA equivalence.
Moreover, EA equivalence is a particular case of CCZ equivalence (see \cite{CCZ, Book}).
Notice that a CCZ transformation might change the algebraic degree of a function and also its bijectivity, whereas both notions are preserved by the affine equivalence and, when the function is not affine, the algebraic degree is also preserved by the EA equivalence.

In the usual setup of multivariate schemes, the central (secret) map  $F$ is composed with two randomly chosen affine bijections $A_1,A_2$ of $\F m$ and $\F n$ respectively to obtain the public map  $G=A_1\circ F\circ A_2$. Thus, the secret and public key $F$ and $G$ are affine equivalent.
Modifying this construction by using EA equivalence does not provide any improvement. 
Indeed, the difference between EA and affine equivalence is just the addition of an affine transformation. So, most of the attacks that can be performed over schemes using an affine transformation can be easily extended to the case of EA transformation.
Therefore, our investigation will focus on the CCZ transformation. Unfortunately, given a function $F$, it seems not so easy to obtain a random CCZ equivalent function $G$ as we are going to explain next.

\subsection{Towards a random CCZ construction}
Let $F$ and $G$ be two CCZ equivalent functions as in Definition \ref{def:eq}.
We can write the affine bijection $\mathcal A:\F n\times\F m\rightarrow\F n\times\F m$ as
$$\mathcal A(z,w)=\mathcal L (z,w)+(a,b),$$
with $a\in\F n$ and $b\in\F m$ and $\mathcal{L}:\F n\times\F m\rightarrow\F n\times\F m$ linear bijection.
Thus,  $\mathcal L$ maps the graph of $F$ to the graph of $G'$, with $G'(z)=G(z+a)+b$.
Hence, up to a translation of the input and the output, we can consider directly the linear bijection $\mathcal L$.
Then, we can write $\mathcal L$ as a matrix composed by four linear maps,
$$\mathcal L(z,w)=\begin{bmatrix}
    A_1 & A_2 \\ A_3 & A_4
\end{bmatrix}\cdot \begin{bmatrix}
    z\\w
\end{bmatrix} = \begin{bmatrix}
    A_1(z)+A_2(w)\\ A_3(z)+A_4(w)
\end{bmatrix} = \begin{bmatrix}
    L_1(z,w)\\ L_2(z,w)
\end{bmatrix}.$$
Recall that $\mathcal L(z,F(z))=(z',G(z'))$.
Set $$F_1(z)=L_1(z,F(z))=A_1(z)+A_2(F(z))$$
and  $$F_2(z)=L_2(z,F(z))=A_3(z)+A_4(F(z)).$$
Clearly, $F_1$ has to be a bijection and $G=F_2\circ F_1^{-1}$.
Notice that both $F_1$ and $F_2$ have degree at most the degree of $F$, while the bound on the degree of $G$ depends also on the degree of the inverse of $F_1$.

With this notation, in order to generate a random function $G$ in the CCZ class of $F$, we need to construct:
\begin{compactenum}
    \item A random linear map $L_1:\F n\times\F m\rightarrow\F n$ such that $L_1(z,F(z))$ is a bijection;
\item A random linear map $L_2:\F n\times\F m\rightarrow\F m$ such that $\begin{bmatrix}
    L_1\\L_2
\end{bmatrix}$ is a bijection over $\F {n+m}$.
\end{compactenum}

Clearly, $L_1$ (and consequentially $L_2$) strongly depends on the choice of the initial function $F$.
So, differently from the affine and the EA equivalence, it appears to be not easy to provide a general construction method for a random CCZ equivalent function.


\section{The twisting}\label{sec:twisting}
To study a possible way to construct a random CCZ equivalent map,
we consider a particular instance of CCZ equivalence, introduced for the case $q=2$ by Canteaut and Perrin with the name of twisting.
Indeed, in \cite[Theorem 3]{CanPer} the authors showed that any two CCZ equivalent functions are connected via the following 3 steps: EA transformation, $t$-twisting, EA transformation.

\subsection{Definition of twisting}
\label{subsec:twistingdef}

We recall the definition of $t$-twisting from \cite{CanPer}, generalized to any finite field $\fF$.

\begin{definition}\label{def:ttwist}
For $\ell$ a positive integer, we denote with $I_\ell$  the $\ell\times\ell$ identity matrix.
Given $0\le t\le \min(n,m)$, we say that two functions $F,G:\F n \rightarrow\F m$  are \emph{equivalent via $t$-twist (or $t$-twisting)} if $\mathcal G_G=M_t(\mathcal G_F)$, where $M_t$ the $(n+m)\times(n+m)$ matrix
$$M_t=\begin{bmatrix}
    0 & 0 & I_t &0\\ 0& I_{n-t} &0&0\\
    I_t &0&0&0 \\ 0&0&0&I_{m-t}
\end{bmatrix}.$$ It holds $M_t=M_t^T=M_t^{-1}$, so this is an equivalence relation.
\end{definition}
Assume that $F,G:\F n\rightarrow\F m$ are equivalent via $t$-twist.
Then, we can split the input and the output of the function $F$ in the first $t$ entries and the remaining $n-t$ (resp.\ $m-t$) entries for input  (resp.\ for output).
That is, for $x\in\F t, y\in\F{n-t}$ we write
\begin{equation}\label{eq:F(T,U)}
F(x,y)=(T(x,y),U(x,y))=(T_y(x),U_x(y))
\end{equation}
with $T:\F t\times\F{n-t}\rightarrow\F{t}$ and $U:\F t\times\F{n-t}\rightarrow\F{m-t}$.
Then, we can write \begin{align}\label{eq:MtG_F}
    M_t(
   (x,y),F(x,y) )=&M_t\cdot\begin{bmatrix}
        x\\y\\T(x,y)\\U(x,y)
    \end{bmatrix}
    =\begin{bmatrix}
        T(x,y)\\y\\x\\U(x,y)
    \end{bmatrix}.
\end{align}
To clarify the notation, here and in
the rest of the paper we identify $x$ with the $t$ variables $x_1,\dots, x_t$, and $y$ with the $n - t$
variables $y_1,\dots, y_{n-t}$. 

From \eqref{eq:MtG_F} we obtain immediately the following important result.

\begin{theorem}\label{rem:Tinvertible}\label{rem:Tdegree}
    Let $F$ be as in \eqref{eq:F(T,U)} and $M_t$ be as in Definition \ref{def:ttwist}. If $T_y:\fq^t\to\fq^t$ is a bijection for any $y\in\fq^{n-t}$, then $M_t(\mathcal G_F)$ is a graph of a function $G$ (i.e. $\mathcal G_G=M_t(\mathcal G_F)$) given by $G(x,y)=(T_y^{-1}(x),U(T_y^{-1}(x),y))$. Moreover, if the degree of $T_y^{-1}(x)$ is $d$ (i.e.\ the algebraic degree with respect to both $x$ and $y)$ and the algebraic degree of $U$ is $d'$, then the degree of $U(T_y^{-1}(x),y)$ is at most $d\cdot d'$.
\end{theorem}

\subsection{Twisting of quadratic functions}
\label{subsec:F quadr}
As mentioned in the introduction, the majority of multivariate cryptographic schemes appearing in the literature deal with quadratic functions.
Therefore, we consider the case when the central map is a quadratic function and we present a preliminary study of the twisting for quadratic maps.

Let $F$ be a quadratic function admitting a $t$-twist. We keep the notation introduced in \S\ref{subsec:twistingdef} and we write $F(x,y)=(T(x,y),U(x,y))$. 
Since $F$ is quadratic, both functions $T$ and $U$ have degree at most 2.
By Theorem~\ref{rem:Tinvertible}, we require that for every $ y\in\F{n-t}$, $T(x, y)=T_{y}(x)$ is invertible.
Then, the bijection $T_y(x)$ is either affine or quadratic.
We deal with both cases separately.

\subsubsection{$T$ affine}
We assume that $T$ is affine, that is $T(x,y)=\ell(x)+\phi(y)$ with $\ell:\F t\rightarrow\F t$ linear bijection and $\phi:\F{n-t}\rightarrow\F t$ affine transformation. In this case, the map is easily invertible as $T_y^{-1}(x)=\ell^{-1}(x)-\ell^{-1}(\phi(y))$.
    However, this map
     will produce a function $G$ that is EA equivalent to $F$.
     Indeed, $G=F\circ A$, with $A$ the affine bijection of the form  $A(x,y)=(\ell^{-1}(x)-\ell^{-1}(\phi(y)),y)$.
     So, we are not interested in this case.

\subsubsection{$T$ quadratic}
The general case $T_y(x)$ quadratic 
    is quite difficult to analyse.
  By Theorem~\ref{rem:Tinvertible}, to apply a CCZ transformation, we need the map $T(x,y)=T_y(x)$  to be invertible for every fixed $y$.
A possible way to achieve this is to consider $T(x,y)$ as a function of the form
\begin{equation}\label{eq:T=l(x)+q(y)}
    T(x,y)=\ell(x)+\q(y),
\end{equation} with $\ell:\F t\rightarrow\F t$ a linear bijection and $\q:\F{n-t}\rightarrow\F t$ a quadratic function.
In this case, the inverse of $T_y(x)$ has the form
$T_y^{-1}(x)=\ell^{-1}(x)-\ell^{-1}(\q(y))$.
Indeed, we have 
\[
T(\ell^{-1}(x)-\ell^{-1}(\q(y)),y)=\ell(\ell^{-1}(x)-\ell^{-1}(\q(y)))+\q(y)=x.
\]
With $T$ as in \eqref{eq:T=l(x)+q(y)}, the map $T_y^{-1}(x)$ has degree at most two. Thus, Theorem~\ref{rem:Tdegree} implies that the map $G$ constructed with the $t$-twist has degree at most four.

\section{Multivariate CCZ Scheme}\label{sec:proposal}
We propose a scheme where the central map is hidden by an application of a CCZ transformation.
We present two versions of this proposal at increasing level of details. First, we present a generic instance of the scheme which can be applied to any central map admitting a $t$-twist (\S\ref{subsec:genproposal}). 
Then, we present some restrictions on the choice of the quadratic secret map used (\S\ref{sec:proposalquadratic} and \S\ref{subsec:onU}), summarized in a concrete instance with a specific choice for the central map (\S\ref{sec:pesto}). Finally, after conducting a preliminary security analysis in the following sections, we will present concrete parameters for this proposal in \S\ref{sec:parameters}.

\subsection{Generic CCZ Scheme}\label{subsec:genproposal}
As ``central map'' we consider $F:\F n\rightarrow\F m$  a function admitting a $t$-twist for an integer $1\le t\le\min(n,m)$. We remove the value $t=0$ since no modification is obtained in this way.
With the usual notation, we write $F$ as
$$F(x,y)=(T(x,y),U(x,y))$$
with $x\in\F t$, $y\in\F{n-t}$, $T:\F t\times\F{n-t}\rightarrow\F t$ such that $T(x,y)=T_y(x)$ is invertible in $x$ for every possible $y$,  and $U:\F t\times\F{n-t}\rightarrow\F{m-t}$.
The equivalent function $G$ has then the form $$G(x,y)=(T_y^{-1}(x),U(T_y^{-1}(x),y)).$$
Finally, we construct the public map as 
$G_{pub}=A_1\circ G\circ A_2,$
for $A_1,A_2$ random affine bijections of $\F m$ and $\F n$ respectively.
In the secret key, we need to store
the information needed to invert the map $G_{pub}$, consisting of 
$\langle A_1,A_2,T,U\rangle$. Equivalently, we can directly consider $\langle A_1^{-1},A_2^{-1},T,U\rangle$.

In the following, we describe in more details the steps to use this pair of public and secret keys for an encryption scheme and for a signature scheme.

\subsubsection{Proposal as an encryption scheme}

In an encryption scheme, a sender encrypts a message by using the public key.
The receiver recovers the original message by knowing the secret key.\\

\emph{Encryption}.
The sender encrypts the message $z\in \F n$ by evaluating the public key $G_{pub}$ and sends 
$c=G_{pub}(z)\in\F m$.\\

\emph{Decryption}. From $c\in \F m$, the receiver  has to compute its preimages, i.e.\ the set of solutions  $\Bar z$ such that $G_{pub}(\Bar z)=c$. This can be obtained in the following way.

\begin{enumerate}
    \item Given $c$, compute $c'=A_1^{-1}(c)$.
    To simplify the notation, for $c=G_{pub}(z)$, call $z'=A_2(z)$, so $c'=G(z')$.
    Moreover,  write $c'$  as $(c'_T,c'_U)\in\F t\times\F{m-t}$ and
    similarly  consider $z'=(x,y)\in\F t\times\F{n-t}$.
    So it holds \begin{align*}
        c'_T=&T_y^{-1}(x),\\
        c'_U=&U(T_y^{-1}(x),y),
    \end{align*}
    implying 
    $c'_U=U(c'_T,y)$.
    \item From $c'_U=U(c'_T,y)$, find the set of possible solution $\mathcal Y=\{y\in\F{n-t} : c'_U=U(c'_T,y) \}$.
    \item For $\Bar y\in\mathcal Y$, compute $\Bar x=T_{\bar y}(c'_T)$.
    \item For every possible pair of solutions $(\Bar x,\Bar{y})$, compute $\Bar z=A_2^{-1}(\Bar x,\Bar{y})$.
\end{enumerate}
In case in the decryption phase multiple solutions are obtained, we need to use  special techniques (for example hash functions or
redundancy in the plaintext) to make the decryption unique, see \cite{Dingch3}.

\subsubsection{Proposal as a signature scheme}
In a signature scheme,  a sender produces a signature for a document, using the knowledge of the secret information.
The receiver checks the validity of the signature for the received document.\\
Consider a hash function $\mathcal H:\{0,1\}^*\rightarrow\F m$.\\

\emph{Signature}.
Given a document $d\in\{0,1\}^*$, the sender wants to create a valid signature, knowing the secret information.
\begin{enumerate}
    \item Compute the hash value $w=\mathcal H(d)$. Compute $w'=A_1^{-1}(w)$ and write it as $w'=(w_T,w_U)\in\F{t}\times\F{m-t}$.
    \item  We need to find $x\in\F t$ and $y\in\F{n-t}$ such that $G(x,y)=w'$.
    So $T_y^{-1}(x)=w_T$ and $U(T_y^{-1}(x),y)=w_U$.
    \item Solve $U(w_T,y)=w_U$, pick randomly one of the solutions and call it $\Bar{y}$.
    \item Solve $T_{\Bar y}^{-1}(x)=w_T$, that is compute $\Bar{x}=T_{\Bar y}(w_T)$.
    \item Given the solution $(\Bar{x},\Bar{y})$, compute $(x',y')=A_2^{-1}(\Bar{x},\Bar{y})$.
    \item Output $(x',y')$ as signature.
\end{enumerate}

\emph{Verification}. The receiver wants to verify that the signature $(x',y')$ is valid for the document $d$.
\begin{enumerate}
    \item Compute $w=\mathcal H(d)$.
    \item Check that $G_{pub}(x',y')=w$.
\end{enumerate}
\paragraph{Correctness of the procedures}
For the signature scheme, the correctness is verified by the following relation,
\begin{align*}
    G_{pub}(x',y')=& A_1\circ G\circ A_2(x',y')= A_1\circ G(\Bar{x},\Bar{y})\\
    =&A_1\circ (T_{\Bar y}^{-1}(\Bar x),U(T_{\Bar y}^{-1}(\Bar x),y))=A_1 (w_T,w_U) = w.
\end{align*}
The correctness of the encryption scheme $G_{pub}(\Bar m)=c$ is verified in the same way.

\paragraph{Conditions for signature scheme and encryption scheme}
For both encryption scheme and signature scheme, we need $T(x,y)=T_y(x)$ to be an invertible function with respect to $x$, for every possible value of $y$. 

Regarding the function $U(x,y)$, different conditions must be satisfied.
If we want to use $F$ as a central map for a signature scheme, we need that for any possible $(a,b)\in\F{t}\times\F{m-t}$, the system $U(a,y)=b$ always has a solution.
Instead, for the encryption scheme, we need the system to have a few solutions. Indeed, in this case, the receiver has to consider all possible pre-images of the map $G_{pub}$ and find the correct one.

\subsection{Our proposal}\label{subsec:ourproposal}
Assume that $F:\F n\rightarrow\F m$ is a quadratic function which admits a $t$-twist, so
$F(x,y)=(T(x,y),U(x,y))$ where both $T$ and $U$ have at most degree 2.
In what follows, we present our proposal for the choice of the maps $T$ and $U$.

\subsubsection{The choice of $T$}\label{sec:proposalquadratic}
Given the analysis presented in \S\ref{subsec:F quadr}, we choose the map $T$ as in Equation~\eqref{eq:T=l(x)+q(y)}, which allows us to easily invert $T_y(x)$.
Therefore, we consider $T(x,y)=\ell(x)+\q(y),$
where $\ell:\F t\rightarrow\F t$ is a linear invertible tranformation and $\q:\F{n-t}\rightarrow\F t$ is a quadratic map.
We show in the following that we do not loose generality by restricting to $\ell(x)=x$ the identity map.
\begin{remark}
The $t$-twisted map $G$ and the public map $G_{pub}$ are $$G(x,y)=(T_y^{-1}(x),U(T_y^{-1}(x),y))=(\ell^{-1}(x)-\ell^{-1}(\q(y)),U(\ell^{-1}(x)-\ell^{-1}(\q(y)),y))$$
and 
$G_{pub}=A_1\circ G\circ A_2,$
with $A_1,A_2$ affine bijections of $\F m$ and $\F n$ respectively.
    We  can always write the affine bijection $A_2$ as $A_2=L\circ A_2'$ with
$L(x,y)=(\ell(x),y)$
and $A_2'=L^{-1}\circ A_2$.
Then, we set $G'=G\circ L$, so we have $G_{pub}=A_1\circ G'\circ A_2'$ with 
\begin{align*}
    G'(x,y)=&
    G\circ L(x,y)
    =
    (x-\ell^{-1}(\q(y)),U(x-\ell^{-1}(\q(y)),y)).
\end{align*}
This means that, since $A_2$ is chosen at random, we can assume without loss of generality that $\ell(x)=x$.
\end{remark}

Given the previous consideration, our choice is
 $$T(x,y)=x+\q(y),$$ leading to a $t$-twist of the form
$G(x,y)=(x-\q(y),U(x-\q(y),y))$.
Observe that with this choice of $T$, to compute $\Bar{x}=T_{\Bar{y}}(w_T)$ simply corresponds to computing $\Bar x=w_T+\q(\bar y)$.

\subsubsection{The choice of $U$}\label{subsec:onU}
Now, we propose a possible choice for the quadratic map $U$.
We recall that we need $U(x,y)$ to be such that, fixed $x$, it is easy to get the preimages (or a preimage) with respect to $y$. A possible way to achieve this is to use Oil and Vinegar (OV) maps.
Therefore, chosen a parameter $s$ with $0\le s\le n-t$,
we propose to construct the map $U$ as a system of OV equations
with $t+s$ vinegar variables and $n-t-s$ oil variables.
To be more specific, in this proposal $U$ consists of a system of $m-t$ equations
of the form
\begin{equation}\label{eq:OVmaps}
  f^{(i)}=\sum_{j,k\in V}\alpha_{jk}^{(i)}z_jz_k+\sum_{j\in V,k\in O}\beta_{jk}^{(i)}z_jz_k+\sum_{j\in V\cup O}\gamma_j^{(i)}z_j+\delta^{(i)}
\end{equation}
with $\{z_j : j\in V\}=\{x_1,\ldots,x_t,y_1,\ldots,y_s\}$ and $\{z_j : j\in O\}=\{y_{s+1},\ldots,y_{n-t}\}$ respectively the sets of vinegar  and oil variables, and with
the coefficients $\alpha_{jk}^{(i)},\beta_{jk}^{(i)},\gamma_j^{(i)},\delta^{(i)}$ randomly chosen over $\fF$.
Notice that, fixed the vinegar variables, the system is linear in the oil variables, hence it is easy to solve, for example with a simple Gaussian reduction.
The legitimate user can get the preimages of $U$ with respect to $y$ (fixed $x$) using classical techniques from OV systems \cite{Ding2009,Ding2020}.
Notice that the easiest system is obtained for $s=0$, however from the analysis presented in \S\ref{sec:securityanalysis}, this is not a good choice.

\subsubsection{\texttt{Pesto}}\label{sec:pesto}
We sum up all the previous choices in the following definition: the scheme \texttt{Pesto}.
\begin{definition}[\texttt{Pesto} scheme]\label{def:Pesto}
Fix positive integer parameters $n,m,t,s$ with $t\le \min(n,m)$ and $s\le n-t$, consider the following maps:\begin{itemize}
    \item $\q:\F{n-t}\rightarrow\F t$ random quadratic map (so $T(x,y)=x+\q(y)$);
    \item $U:\F t\times\F{n-t}\rightarrow\F{m-t}$ a system of $m-t$ random OV maps with $x_1,\ldots,x_t,y_1,\ldots,y_s$ vinegar variables and $y_{s+1},\ldots,y_{n-t}$ oil variables as in Equation~\eqref{eq:OVmaps};
    \item $A_1:\F m\rightarrow\F m$ a random affine bijection;
    \item $A_2:\F n\rightarrow\F n$ a random affine bijection.
\end{itemize}
Set $G=(x-\q(y),U(x-\q(y),y))$.
Then the map $G_{pub}=A_1\circ G\circ A_2$ is the public key, and  $\langle A_1,A_2,\q,U\rangle$ constitutes the secret key.  
\end{definition}
In the following remark, we stress the role of $s$.
\begin{remark}
   The amount of possible signatures for a document, or the amount of possible plaintexts for a given ciphertext,  depends on the value of $s$. 
   \begin{itemize}
    \item For an encryption scheme, we need to find all possible solutions of $c'_U=U(c'_T,y)$.
    Here we need to try all possible values for $y_1,\ldots,y_s$ (that is, $q^s$ possibilities) and then solve a linear system of $m-t$ equations in $n-t-s$ variables.
    \item For a signature scheme, we need to find only one solution of $w'_U=U(w'_T,y)$.
    Hence, we pick random values for $y_1,\ldots,y_s$
    and then we solve a linear system of $m-t$ equations in $n-t-s$ variables.
    If the system does not have a solution, we pick other random values for $y_1,\ldots,y_s$.
\end{itemize}
Hence, when choosing the parameters $n,m,t,s$, one has also to consider if the linear system to solve, $m-t$ equations in $n-t-s$ variables, should be:\begin{itemize}
    \item determined with high probability ($m-t=n-t-s$, that is $s=n-m$);
    \item overdetermined ($m-t>n-t-s$, that is $s>n-m$);
    \item underdetermined ($m-t<n-t-s$, that is $s<n-m$).
\end{itemize}
\end{remark}

\subsubsection{Toy Example}
\label{toyexample}
We provide here a toy example of \texttt{Pesto} over $\mathbb{F}_5$.
We take $n=5$, $m=4$, $t=2$, and $s=1$.
Therefore we have the following set of variables: $x=\{x_1,x_2\}$ and $y=\{y_1,y_2,y_3\}$.
To define the map $T(x,y):\mathbb{F}_5^2\times\mathbb{F}_5^4\rightarrow\mathbb{F}_5^2$, we need to define a quadratic map $\q:\mathbb{F}_5^4\rightarrow\mathbb{F}_5^2$.
Set $$\q(y)=\begin{bmatrix}\small
      y_1^2 + 2y_1y_2 + 4y_2^2 + 3y_2y_3 + y_2 + 3y_3^2 + 4\\ \\
    3y_1^2 + 3y_1y_2 + 2y_1y_3 + 2y_2y_3 + 2y_2 + y_3^2 +
        2y_3
\end{bmatrix},$$
therefore
 $T(x,y)=x+\q(y)$.
The map $U$ is an OV system of 2 equations with $x_1,x_2,y_1$ vinegar variables and $y_2,y_3$ oil variables.
Hence we consider
$$U(x,y)=\begin{bmatrix}\small
     x_1^2 + 2x_1x_2 + 3x_1y_1 + x_1y_2 + 4x_1y_3 + 2x_1 +
        x_2^2 + x_2y_1+ x_2y_2\\
         + 3x_2y_3 + 3x_2
        + 2y_1^2 +
        2y_1y_2 + 2y_1y_3 + y_1 + 4y_2 + y_3\\ \\
    x_1x_2 + x_1 + 4x_2^2 + 2x_2y_2 + 3x_2y_3 + 3x_2 +
        2y_1y_3 + 3y_1 + 3y_3 + 1
\end{bmatrix}.$$

Having constructed these maps, we have that $G(x,y)$ has the following form
$$
\begin{bmatrix}\small
    x_1 + 4y_1^2 + 3y_1y_2 + y_2^2 + 2y_2y_3 + 4y_2 + 2y_3^2 +
        1\\ \\
    x_2 + 2y_1^2 + 2y_1y_2 + 3y_1y_3 + 3y_2y_3 + 3y_2 +
        4y_3^2 + 3y_3\\ \\
    x_1^2 + 2x_1x_2 + 2x_1y_1^2 + x_1y_1y_3 + 3x_1y_1 +
        2x_1y_2^2 + 2x_1y_3^2 + 4x_1 + x_2^2 + 2x_2y_1^2 +
        x_2y_1y_3  \\
        + x_2y_1+ 2x_2y_2^2 + 2x_2y_3^2 + 4x_2y_3
        + y_1^4 + y_1^3y_3 + 4y_1^3 + 2y_1^2y_2^2 + y_1^2y_2 +
        y_1^2y_3^2 + y_1^2y_3 + 3y_1^2\\
        + y_1y_2^2y_3 +
        3y_1y_2^2 + 2y_1y_2y_3 + 4y_1y_2 + y_1y_3^3 +
        2y_1y_3^2 + 4y_1 + y_2^4 + 2y_2^2y_3^2 + 2y_2y_3^2 \\
        +        3y_2y_3 + y_2+ y_3^4 + y_3^3 + y_3^2 + 3\\ \\
    x_1x_2 + 2x_1y_1^2 + 2x_1y_1y_2 + 3x_1y_1y_3 +
        3x_1y_2y_3 + 3x_1y_2 + 4x_1y_3^2 + 3x_1y_3 + x_1 +
        4x_2^2 + 4x_2y_1y_2\\
        + 4x_2y_1y_3 + x_2y_2^2 +
        x_2y_2y_3 + 4x_2y_3^2 + 2x_2y_3 + 4x_2 + 4y_1^4 +
        y_1^3y_2 + 4y_1^2y_2^2 + y_1^2y_2y_3 + 2y_1^2y_2 +\\
        y_1^2y_3 + 2y_1^2 + 2y_1y_2^3 + 4y_1y_2^2y_3 +
        4y_1y_2^2 + 3y_1y_2y_3^2 + 3y_1y_2y_3 + y_1y_2 +
        2y_1y_3^3+ y_1y_3^2\\
         + 4y_1y_3 + 3y_1 + 3y_2^3y_3 +
        3y_2^3 + y_2^2y_3^2 + 4y_2^2y_3 + 3y_2y_3^2 + 3y_2y_3 +
        y_2 + 2y_3^4 + 4y_3^3 + 3y_3^2 + 2
\end{bmatrix}.$$

Finally, we consider the following affine bijections of $\mathbb{F}_5^5$  and $\mathbb{F}_5^4$
$$A_2(x,y)=\begin{bmatrix}
    1 &4& 3& 2& 1\\
2& 0& 1& 1& 4\\
3& 2& 2& 0& 2\\
1& 2& 2& 2& 3\\
2& 3& 4& 4& 2
\end{bmatrix}\begin{bmatrix}
    x_1\\x_2\\y_1\\y_2\\y_3
\end{bmatrix}+\begin{bmatrix}
    2\\ 1\\ 3\\ 2\\ 2
\end{bmatrix},\qquad A_1(z)=\begin{bmatrix}
   2& 3& 2& 1\\
4& 2 &3& 1\\
1& 2 &1& 3\\
1& 4& 3 &1
\end{bmatrix}\begin{bmatrix}
    z_1\\z_2\\z_3\\z_4
\end{bmatrix}+\begin{bmatrix}
   1\\ 0\\ 0\\ 4
\end{bmatrix}.$$
All the mentioned functions, namely $\q,U,A_1,A_2$, were randomly generated with the help of the MAGMA software \cite{magma}.
The public map $G_{pub}=A_1\circ G\circ A_2$ consists of $4$ dense polynomials of degree 4. 
Given their sizes, we opt for not reporting them here.


\section{Computational remarks}\label{sec:computationalrmks}
From now on, we focus on our proposal \texttt{Pesto} of Definition~\ref{def:Pesto}.
We study the form of the secret and public key, their size, and the cost of the procedure.

\subsection{The form of $G$ and $G_{pub}$}\label{sec:comments}
Since the map $T(x,y)=x+\q(y)$ and its inverse  $T_y^{-1}(x)=x-\q(y)$ have degree $2$, by Theorem~\ref{rem:Tdegree} the degree of the $t$-twisted map $G$ is at most $4$.
We analyse the monomials appearing in $G$ (and in $G_{pub}$) more closely.
\begin{remark}
As stated above, the first $t$ coordinates of $G$ are those from $T_y^{-1}(x)=x-\q(y)$, which have degree 2 and consist of terms of the form $x_i, y_i, y_iy_j$.
Now, we consider the last $m-t$ coordinates of $G$, corresponding to $U(T_y^{-1}(x),y)=U(x-\q(y),y)$.
Since $U$ consists of OV quadratic polynomials $f^{(i)}$ of the form of Equation~\eqref{eq:OVmaps}, in $U(x,y)$ we can find terms of the  form $x_i, y_i, x_ix_j, x_iy_j, y_iy_j$.
By evaluating $x\mapsto x-\q(y)$, the variable $x_i$ might produce terms of the form $x_j$, $y_j$ and $y_jy_k$.
Therefore, potentially, in the last $m-t$ coordinates of $G$ we can have terms of the form
$x_i$, $y_i$, $x_ix_j$, $x_iy_j$, $y_iy_j$, $x_iy_jy_k$, $y_iy_jy_k$, $y_iy_jy_ky_l$.
Notice that, even if we replace $U$ with a dense quadratic map, the possible terms in $G$ have the same form.
To sum up, even if the degree of $G$ is up to 4, the terms of degree 4 involve only variables in $y$, while terms of degree 3 involve only variables in $y$ or 2 variables in $y$ and 1 variable in $x$.
\end{remark}

The public map $G_{pub}=A_1\circ G\circ A_2$ consists of dense polynomials of degree up to 4, since the random map $A_2$ will remove the above-mentioned restrictions.
However,  the affine transformation keeps invariant the amount of components of a fixed degree, so to analyse the possible degrees of the components of $G_{pub}$, we can directly consider the map $G(x,y)=(x-\q(y), U(x-\q(y),y))$.
So, $G_{pub}$ has at least $q^t-1$ quadratic components, since the components $\lambda\cdot G$ with $\lambda=(v,0_{m-t})$ with nonzero $v\in \F t$ are quadratic.
Moreover, we note that in all the systems constructed during the preparation of this paper, $U(x - \q(y), y)$ does not contain quadratic components. 
This is expected since the polynomials in $U$ and $\q$ are quadratic polynomials with random coefficients. In \S\ref{sec:isolatingquadratic}, we will show how this can be leveraged to propose a reduction in the key sizes.

\subsection{Dimensions of the keys}\label{subsec: dim-of-key}
In the following, we provide an analysis of the sizes of the public and private keys. Furthermore, in Table~\ref{tab:sizekeys} we will present a collection of concrete values of the key sizes for a selection of parameters.

\begin{proposition}\label{prop:pestokeysize}
    Consider a \texttt{Pesto} scheme as proposed in Definition \ref{def:Pesto}.
    Then the public key consists of $m\cdot{n+4\choose 4}$ coefficients over $\fF$, and the secret key consists of $m^2+m+n^2+n+t{n-t+2\choose2}+(m-t){t+s+2\choose2}+(m-t)(n-t-s)(t+s+1)$ coefficients over $\fF$.    
\end{proposition}
\begin{proof}
    Recall that a polynomial of degree $r$ in $n$ variables with coefficients over $\fF$  has $M_n(r)={n+r\choose r}$ terms.
Notice that when $q\le r$, we can use the field equations $x_i^q-x_i$ to lower the degrees of the polynomials over  $\fF$ and store fewer coefficients.
 $G_{pub}$ is a system of $m$ equations of degree at most 4 in $n$ variables, where
each equation has $M_n(4)={n+4\choose 4}$ possible terms.
Hence, for $G_{pub}$ we need to specify $m\cdot M_n(4)$ coefficients over $\fF$.

Regarding the secret key, we want to store the information to recover $A_1,A_2,\q,U$.
\begin{itemize}
    \item $A_1$ and $A_2$ are  affine bijections over $\F m$ and $\F n$ respectively: They correspond to  invertible $m\times m$ and $n\times n$ matrices over $\fF$ plus a constant vector.
    So this means $m^2+m+n^2+n$ coefficients over $\fF$.
    \item $\q:\F{n-t}\rightarrow\F t$ quadratic: It corresponds to $t$ quadratic equations in $n-t$ variables.
    So,  we need to store $t\cdot M_{n-t}(2)$ coefficients over $\fF$. 

    \item $U:\F{n}\rightarrow\F{m-t}$  a quadratic OV system of $m-t$ equations with $t+s$ vinegar variables and $n-t-s$ oil variables.
    We can see each equation as a quadratic map in $t+s$ variables, $M_{t+s}(2)$ terms, plus a map with oil variables multiplied by a vinegar variable or multiplied by 1, $(n-t-s)(t+s+1)$ terms.
    Thus, we need to store $(m-t)\cdot\big( M_{t+s}(2)+(n-t-s)(t+s+1)\big)$ coefficients over $\fF$.
\end{itemize}
\end{proof}

\subsection{Computational cost of the procedure}
In this subsection, we present a tentative analysis of the computational cost of the proposed multivariate scheme. Notice that the application of this procedure reduces to the evaluation of polynomials of degree at most four and computing the solution of linear systems. 
Considering the direct evaluation of a polynomial, we have the following estimates. 

\begin{itemize}
    \item To evaluate affine polynomials we perform $\mathfrak m_1(n)=2{n+1\choose1}-n-2=n$ multiplications and ${n+1\choose1}-1=n$ additions.
    \item To evaluate quadratic polynomials we perform $\mathfrak m_2(n)=2{n+2\choose2}-n-2=n(n+2)$ multiplications and ${n+2\choose2}-1$ additions.
    \item To evaluate quartic polynomials we perform $\mathfrak m_4(n)=2{n+4\choose4}-n-2$ multiplications and ${n+4\choose4}-1$ additions.
\end{itemize}
This is clearly an upper bound, since more efficient techniques might be used, see for example \cite{BallicoEliaSala}.
We indicate with $\mathfrak M(r,n)$ the number of multiplications needed to solve a linear system of $r$ equations in $n$ variables.
Recall that $\mathfrak M(r,n)$ is at most of the order $rn\cdot\min(r,n)$, see for example \cite[Appendix B]{Gauss-red}.

In the following, we present an estimate on the cost of applying the proposed procedure in terms of multiplications, since to multiply is more expensive than to add.
\begin{proposition}
    Consider a signature scheme based on the \texttt{Pesto} scheme as in Definition~\ref{def:Pesto}.
    Set $\mathfrak m_i(r)$ and $\mathfrak M(r,k)$ as defined before.
    Excluding the computation of the hash value, the computational cost, in terms of multiplications over $\fF$, to verify the validity of a signature is $m\cdot\mathfrak m_4(n)$, whereas the computational cost to produce a valid signature is
    $m\cdot\mathfrak m_1(m)+t\cdot\mathfrak m_2(n-t)+n\cdot\mathfrak m_1(n)+(m-t)\left(\mathfrak m_2(t)+\mathfrak m_2(s)\right)+\mathfrak M(m-t,n-t-s)$.
\end{proposition}
\begin{proof}

Assume that the scheme was already initialized, a secret key and the corresponding public key were already constructed.
To verify the validity of a signature $\mathfrak s\in\F n$ we need to evaluate $G_{pub}(\mathfrak s)$.
$G_{pub}$ is a system of $m$ equations in $n$ variables of degree at most 4.
This corresponds to a total of $m\cdot\mathfrak m_4(n)=m\left(2{n+4\choose4}-n-2\right)$ multiplications.
\\
To create a valid signature for $w\in\F m$, we have to perform the following.\begin{enumerate}
    \item Compute $w'=A_1^{-1}(w)$. Since $A_1^{-1}$ is an affine bijection, this means $m\cdot\mathfrak m_1(m)=m^2$ multiplications.
    \item Evaluate $U(w_T',y)$.
    This is an evaluation of quadratic polynomials in $t$ variables and corresponds to $(m-t)\cdot\mathfrak m_2(t)=(m-t)t(t+2)$ multiplications.
    \item  Find a preimage of $U(w_T',y)=w_U'$.
    \item For one preimage $\Bar{y}$ found, evaluate $\Bar x=w_T'+\q(\Bar y)$.
    Since $\q$ is a system of $t$ quadratic equations in $n-t$ variables, this means $t\cdot\mathfrak m_2(n-t)=t\cdot(n-t)(n-t+2)$ multiplications.
    \item From the solution pair  $(\Bar x,\Bar y)$, compute $A_2^{-1}(\Bar x,\Bar y)$. This corresponds to $n\cdot\mathfrak m_1(n)=n^2$ multiplications.
\end{enumerate}

We are left with analysing the cost for the third step, where we have to find a preimage of $U(w_T',y)=w_U'$.
Hence, we need to set $y_1,\ldots,y_s$ to random values, evaluate those variables and then solve a linear system.
The cost of evaluating $s$ variables corresponds to $(m-t)\cdot\mathfrak m_2(s)=(m-t)s(s+2)$ multiplications.
The cost of solving the final system of $m-t$ equations in $n-t-s$ variables is $\mathfrak M(m-t,n-t-s)$ multiplications.
\end{proof}

\begin{remark}
If we consider an encryption scheme, the cost of encrypting a message and decrypting a valid ciphertext can be deduced from the analysis just reported. The fundamental difference is that we will need to compute the entire set of preimages of $U(c_T',y)=c_U'$.
This means that, in the second step, we need to evaluate $y_1,\ldots,y_s$ in every possible value. Hence, we need to compute $q^s$ evaluations and then solve $q^s$ linear systems.
\end{remark}

\section{Security Analysis}\label{sec:securityanalysis}
In this section, we provide some considerations on the security of the scheme \texttt{Pesto} of Definition~\ref{def:Pesto}. Among them, we consider attacks that have been exploited for the MI cryptosystem (and its generalizations) and we analyse under which conditions it would be possible to extend these also to our system. 

\subsection{The importance of $A_2$}

We start this analysis by presenting a first observation on the security of the scheme, which stresses the importance of the affine bijection $A_2$.
Indeed, if an attacker is able to recover $A_2$, then they can operate as follows.
First, compute $\Bar G=G_{pub}\circ A_2^{-1}$ and isolate its quadratic components: 
$t$ of them are of the form $x_i-\q_i(y)$, for $i=1,\ldots,t$. 
From this, it is rather easy to solve the system.
Suppose a ciphertext $c$ is given, and the attacker wants to find a message $m$ such that $G_{pub}(m)=c$.
Write $G_{pub}(m)=\Bar G\circ A_2(m)=c$. From the isolated quadratic equations, recover $x=\q(y)+c_T$. 
By substituting this into the remaining equations of $\bar G$, solve the system, which has at most $s$ variables appearing in quadratic terms. Notice that this last step is equivalent to what a legitimate user has to do.
Once the solution $(x,y)$ of this system is found, then $m=A_2^{-1}(x,y)$.

In order to perform this attack, the attacker has to completely recover $A_2$ but only partially $A_1$ (it is enough to isolate the quadratic components of the form $x_i-\q_i(y)$).
The cost to isolate these quadratic components will be treated in the following section. 

\subsection{Isolating the quadratic components: the shape of $A_1$}\label{sec:isolatingquadratic}
In this section, we explain a strategy to isolate, within $G_{pub}$, the components that correspond to a (linear combination of) the first $t$ coordinates of $G$. This demonstrates that, in many situations, the public system can be partitioned into two sets of equations and suggests that the external affine transformation $A_1$ can be chosen in a specific form, thereby reducing the size of the secret key.

We try to estimate the cost of isolating the quadratic components of $G_{pub}$.
Let $r\ge t$ be such that the number of quadratic components of  $G_{pub}$ is $q^r-1$.
Notice that when $r=t$, then $U(x-\q(y),y)$ does not have quadratic components. As mentioned in \S\ref{sec:comments},  this situation occurred in all the experiments we conducted.
We consider the list $\Delta$ of all possible terms of degree 3 and degree 4 in $n$ variables appearing in $G_{pub}$.
We have
\begin{align*}
        |\Delta|=&{n+4\choose4}-{n+2\choose2}=(n+2)(n+1)\frac{n^2+7n}{24}.
\end{align*}
We construct the matrix $A$ with $n$ rows and $|\Delta|$ columns, with at position $(i,j)$ the coefficients of the $j$-th term of $\Delta$ in the $i$-th equation of $G_{pub}$.
So, to recover the quadratic equations of $G_{pub}$ it is enough to solve the linear system $x^TA=A^Tx=0_{|\Delta|}$. The set of solutions corresponds to the set of quadratic components of $G_{pub}$. The size of this set is $q^r$.
Recall that the cost of solving a linear system with $n$ equations in $|\Delta|$ variables $\mathfrak M(n,|\Delta|)$ is of the order $O(n^6)$.
Note, however, that this approach allows us to recover the space of quadratic components of $G\circ A_2$, but not the exact quadratic coordinates.

\begin{remark}\label{rmk:A1 reduced}
    The previous observations suggest that a dense affine transformation $A_1$, which completely mixes the first $t$ equations of $G$ with the remaining $m-t$ equations, may not provide additional security, since the first $t$ equations can be recovered at a cost of $O(n^6)$ by an attacker. Therefore, we suggest considering the matrix corresponding to the linear part of $A_1$ in the form $\begin{pmatrix}* & 0 \\ * & *\end{pmatrix}$, with a zero block of size $t\times (m-t)$. This way, we do not fully mix all the equations and we reduce the secret key and public key sizes. With this observation, the secret key size of \texttt{Pesto} from Proposition~\ref{prop:pestokeysize} becomes 
    \[
    m^2-t(m-t)+m+ n^2+n+t{n-t+2\choose2}+(m-t){t+s+2\choose2}+(m-t)(n-t-s)(t+s+1)
    \]
    coefficients over $\fF$ to store;
    the public key size of \texttt{Pesto} from Proposition~\ref{prop:pestokeysize} becomes 
    \[
    (m-t)\cdot{n+4\choose 4}+t\cdot{n+2\choose 2}
    \]
    coefficients over $\fF$ to store.

    In Table~\ref{tab:sizekeys}, we present a comparison of the sizes of secret and public keys between the original version with a generic affine transformation $A_1$ as in Proposition~\ref{prop:pestokeysize} and this modified version. The choice of parameters will be explained in \S\ref{sec:parameters}.
    
\begin{table}[h]
   \centering
   \begin{tabular}{|c|c|c|c||c|c|| c|c|}
   \hline
       $n$& $m$ & $t$& $s$ & amount for $pk$ & amount for $sk$ & amount for $pk_r$ & amount for $sk_r$  \\
       \hline\hline
       27&25&10&2&$\sim\num{786}\cdot10^3$ 
       &\num{7406} &$\sim\num{476}\cdot10^3$ 
       & \num{7256} \\ \hline
       40& 38 & 14 & 2& $\sim\num{515}\cdot10^4$ 
       &\num{21878}& $\sim\num{327}\cdot10^4$ 
       &\num{21542}\\
        \hline
        57&55&20&2&$\sim\num{287}\cdot10^5$ 
        & \num{59041}& $\sim\num{182}\cdot10^5$ 
        & \num{58341}\\
        \hline
   \end{tabular}
   \caption{Amount of coefficients of $\mathbb F_q$ to store, where $pk$ and $sk$ refer to the original system, and $pk_{r}$ and $sk_{r}$ refer to the reduced system with $A_1$ with a zero block.}
   \label{tab:sizekeys}
\end{table}
\end{remark}

\subsection{Using linear structures to analyse $A_2$}
In this section, we present an analysis using linear structures (see Definition \ref{def:lin.struct} below for the precise definition) of some components of $G_{pub}$.
We outline the core idea and explain how, under suitable assumptions, this can be used to recover partial information on the affine transformation $A_2$.
However, at this stage, this analysis does not appear sufficient to fully determine $A_2$ and, consequently, does not seem to compromise the scheme.

The core idea of this analysis is the following.
\begin{enumerate}
    \item $G_{pub}$ has (at least) $q^t-1$ quadratic components with (at least) $q^t$ linear structures in common. Specifically, they are the linear combinations of the first $t$ coordinates of $G\circ A_2$.
    Indeed, for quadratic functions of the form $x_i+\q_i(y)$, and their linear combinations, there are at least $q^t$ linear structures common to all these functions.
    \item If we are able to
isolate these components in $G_{pub}$ (corresponding to the mentioned components of $G$), we take $t$ of them (linearly independent) and then we determine the intersection of the linear structures on these $t$  components of $G_{pub}$, which contains a $t$-dimensional vector space $V$.
\item 
If we are able to identify this subspace, we know that the linear part of $A_2$ maps $V$ into $\F t$. So, we can recover some partial information on this linear part.
Specifically, we know that in $V$ there are $t$ linearly independent elements which form the first $t$ columns of the inverse of the linear part of $A_2$.
\end{enumerate}

We present now in more detail the analysis roughly explained above step by step. We start by recalling the definition and some useful properties of linear structures.

\begin{definition}\label{def:lin.struct}
    Let $f:\F n\rightarrow \fF$ be a function. We say that $a\in\F n$ is a {\em linear structure} of $f$ if the derivative $D_af(x):=f(x+a)-f(x)$ is constant.
\end{definition}

\begin{lemma}
The set of linear structures of a function $f$ forms a vector subspace of $\F n$.
\end{lemma}

 \begin{proof}
  Indeed, $D_{a+b}f(x)=f(x+a+b)-f(x)=f(x+a+b)-f(x+a)+f(x+a)-f(x)=D_bf(x+a)+D_af(x)$.
 Instead, if $a$ is a linear structure of $f$, to show that $\tau a$ is also a linear structure, for any $\tau\in\fF$, we proceed as follows.
 Consider $L$ a linear bijection of $\F n$ such that $L(e_1)=a$, where $e_1$ is the first element of the canonical basis of $\F n$.
 Then, for $g=f\circ L$, $e_1$ is a linear structure of $g$ since $D_bg(x)=g(x+b)-g(x)=f(L(x)+L(b))-f(L(x))=D_{L(b)}f(L(x))$, and for $b=e_1$ the derivative is constant.
 Since we represent functions with polynomials of degrees $<q$ in each variable, this implies that $g$ is of the form
 $g(x_1,\ldots,x_n)=\alpha x_1+h(x_2,\ldots,x_n)$, for $\alpha\in\fF$ and $h:\F{n-1}\rightarrow\fF$. 
 Therefore for any $\tau\in\fF$
 $\tau e_1$ is also a linear structure of $g$, implying
 $D_{\tau e_1}g(x)=D_{L(\tau e_1)}f(L(x))$.
 Since $L(\tau e_1)=\tau L(e_1)=\tau a$, this concludes the proof.
 \end{proof}

\begin{lemma}\label{lemma:linearstructures}
    The multiset of the number of linear structures of the components of $G$ and of $G_{pub}$ is the same.
    Moreover, $a\in\F n$ is a linear structure of $\lambda\cdot G_{pub}$ if and only if $L_2(a)$ is a linear structure of ${L_1^*(\lambda)}\cdot G$
    where $L_1$ and $L_2$ are the linear parts of the affine bijections $A_1$ and $A_2$, and $L_1^*$ is the \emph{adjoint operator} of $L_1$,
that is $x\cdot L_1(y)=L_1^*(x)\cdot y$ for every $x,y$.
\end{lemma}
\begin{proof}
    It is known that functions in a given EA-class  have the same number of components with the same number of linear structures.
    Indeed, for each nonzero $\lambda\in\F m$ there exists a nonzero $\gamma\in\F m$ such that
    $\lambda\cdot G_{pub}(x)=\gamma\cdot G(A_2(x))$. Then  clearly the number of linear structures of $\lambda\cdot G_{pub}$ equals the number of linear structures of $\gamma\cdot G$.\\
    The second statement can be deduced from the following.
    From the definition of $G_{pub}$ and since $A_1$ is affine we have
    $D_a(\lambda\cdot G_{pub})(x)=\lambda\cdot G_{pub}(x+a)-\lambda\cdot G_{pub}(x)=\lambda\cdot A_1\circ G\circ A_2(x+a)-\lambda\cdot A_1\circ G\circ A_2(x)=\lambda\cdot L_1\circ G\circ A_2(x+a)-\lambda\cdot L_1\circ G\circ A_2(x)$.
    Then $D_a(\lambda\cdot G_{pub})(x)=\lambda\cdot L_1(G(A_2(x)+L_2(a))-G(A_2(x)))=\lambda\cdot L_1(D_{L_2(a)}G(A_2(x)))=L_1^*(\lambda)\cdot D_{L_2(a)}G(A_2(x))$.
\end{proof}
Now we explain the three steps.

\underline{Step 1}.
By Lemma~\ref{lemma:linearstructures} we can directly count the number of linear structures for the components of $G$ instead of $G_{pub}$.
Notice that for $\lambda=(\lambda',0)\in\F t\times\F{m-t}$ with any nonzero $\lambda'\in\F t$, the $\lambda$-component of $G$ has the form $G_\lambda(x,y)=\lambda'\cdot (x-\q(y))$.
Then, picking any  element of the form $a=(a',0)\in\F t\times\F{n-t}$,
it holds $G_\lambda(x+a',y)-G_\lambda(x,y)=\lambda'\cdot(x+a'-\q(y)-(x-\q(y)))=\lambda'\cdot a'$.
So, $a$ is a linear structure of the function $G_\lambda$ for any such $\lambda$. 
The number of $a$ of this form is 
  $q^t$ and the number of $\lambda$ of this form is $q^t-1$.
Therefore, 
$G_{pub}$ admits at least $q^t-1$ nonzero components having at least $q^t$ linear structures in common. These common linear structures form a $t$-dimensional vector space $V$.

\underline{Step 2}. To determine $V$, an attacker has to do the following:
\begin{itemize}
    \item isolate the components in $G_{pub}$ corresponding to the $\lambda$-components of $G$ (with $\lambda=(\lambda',0)\in\F t\times\F{m-t}$, for any $\lambda'\in\F t$);
    \item among the isolated components, select $t$ linearly independent, compute the intersection of the linear structures of these components and recover the $t$-dimensional vector space $V$.
\end{itemize}
The desired components of $G_{pub}$ are quadratic; therefore by isolating all the quadratic components of the public function, we are isolating the desired components plus (eventually) the quadratic components of $U(x-\q(y),y)$.
Recall that in all the computational experiments conducted, 
$U(x-\q(y),y)$  did not have components of degree 2.

\underline{Step 3}. We partially recover the linear part of $A_2$ in the following way.
If we are able to determine 
 $V$, then we know that $L_2(V)=\F t\times \{0_{n-t}\}$, where we used the notation $\F t\times \{0_{n-t}\}=\{ (a_1,\ldots,a_t,0,\ldots,0)\in\F n : a_i\in\fF, 1\le i\le t\}$. 
 Equivalently, $V=L_2^{-1}(\F{t}\times \{0_{n-t}\})$.
We observe that if $t=1$, then we can recover the first column of $L_2^{-1}$, up to a multiplication by a nonzero element of $\fF$.
On the other hand, for larger values of $t$, we only know that there exist $t$ linearly independent elements in $V$ which form the first $t$ columns of $L_2^{-1}$ (corresponding to the evaluation of the first $t$ elements of the canonical basis), but finding the precise value of the columns is still not easy.

\subsection{A linearization attack for $s=0$}
We present here an attack which follows the idea of the linearization attack proposed against the MI cryptosystem, see \cite{MI}.
This attack of Patarin \cite{Pat95} works because there exists a bilinear equation in the input $\mathfrak i$ and in the output $\mathfrak o$ of the public map, namely $B(\mathfrak i,\mathfrak o)=0$. This allows first to reconstruct the map $B$ with a relatively small number of input-output pairs, and then given a targeted output, to recover the corresponding input.

First, we consider a \texttt{Pesto} scheme with  $t=1$ (so $x=(x_1)$), $s=0$ and such that $U(x,y)=U(x_1,y)$ does not have  the quadratic term $x_1^2$.
In this case we have $U(x_1,y)=x_1\alpha(y)+\beta(y)$, with $\alpha,\beta:\F{n-1}\rightarrow\F{m-1}$ affine maps.
Hence, for $G(x_1,y)=(c_T,c_U)$ it holds that 
\[
    c_U=U(x_1-\q(y),y)=U(c_T,y)=c_T\alpha(y)+\beta(y).
\]
The above is a bilinear equation in the input $y$ and the output $(c_T,c_U)$ of $G$, implying the existence of a bilinear equation in the input and the output of $G_{pub}$, $B(\mathfrak i,\mathfrak o)=0$.

Now, we consider a generic OV map $U$ with $s=0$ and $t=1$, then in $U$ we  have also the term $x_1^2\delta$, with $\delta\in\F{m-1}$.
Thus we have
$c_U=c_T\alpha(y)+\beta(y)+c_T^2\delta,$
implying that, in order to reconstruct the map for $G$, and so for $G_{pub}$, we need more pairs of (input,output) since the output appears also in quadratic terms.
However, once we have reconstructed the equation $B(\mathfrak i,\mathfrak o)=0$ for $G_{pub}$, we have that given a possible output $\Bar{\mathfrak o}$, the equation $B(\mathfrak i,\Bar{\mathfrak o})=0$ is linear in the input.

This same analysis can be generalized to the case $t>1$ and $s=0$.
Indeed, the map $U$ results to have terms of the form $x_ix_j$, $x_iy_j$, $x_i$, $y_j$.
So, given $G(x,y)=(c_T,c_U)$ we have $m-t$ quadratic equations in $y,c_T,c_U$, with the variables in $y$ appearing only in degree 1.
Hence, also from $G_{pub}$ it is possible to recover an equation $B(\mathfrak i,\mathfrak o)=0$ which is quadratic but with the variables of the input appearing only in degree 1.
Hence, the same attack can be performed.
 On the other hand, if $s>0$, the attack in its current form cannot be performed because the resulting relation $c_U=U(c_T,y)$ will not be linear in the input variables $y$. 
For this reason, we choose $s>0$ in Definition \ref{def:Pesto}.

\subsection{Known attacks on Oil and Vinegar systems}\label{subsec:knownUOVattacks}

The secret map $F$ is formed by two OV systems $T(x,y)$ and $U(x,y)$.
So one may wonder whether the known attacks to OV systems can be performed also to the scheme \texttt{Pesto}.

After the twisting transformation is applied, the first $t$ coordinates (of the map $G$) remain an OV system $x-\q(y)$, while the last $m-t$ coordinates might increase up to degree 4.
So, if one is able to isolate from $G_{pub}$ the components corresponding to $x-\q(y)$, known attacks to OV systems can be applied to these components only.
Therefore, given the well-known  Kipnis–Shamir  attack \cite{KS98}  on  balanced Oil
and Vinegar signature schemes, we recommend to keep the system $x-\q(y)$ unbalanced by setting, for example, $t\approx n/3$.

On the other hand, since the second part of the system $G$ is formed by polynomials of degree up to 4,
we believe that the usual OV attacks cannot be applied directly to the whole map $G_{pub}$.

\subsection{Algebraic attack with Gr\"obner bases}\label{sec:algebraicattack}
In this section, we consider algebraic attacks using Gröbner bases.
The scenario is that an attacker wants to forge a signature, hence to find a preimage of a random element.
We consider a random value $w\in\F m$ (the hash of a document), the goal is to find $v\in\F n$ such that $G_{pub}(v)=w$. 
We can do this by finding a Gr\"obner basis of the polynomial system $G_{pub}$ by the usual strategy (see e.g. \cite{CaminataGorla2021}) with linear-algebra-based-algorithms such as F4 \cite{F4paper}, F5 \cite{F5paper}, XL \cite{XLpaper}, etc.
For a large class of these algorithms, the complexity can be  bounded from above by
\begin{equation}\label{eqGBcomplexity}
O\left(\binom{n+\sd(G_{pub})}{n}^\omega\right)
\end{equation}
where $n$ is the number of variables, $\sd(G_{pub})$ is the solving degree, and $2<\omega<3$. 
In a nutshell, the solving degree represents the highest degree of polynomials that need to be considered during the process of solving the system \cite{CaminataGorla2023}.
In order to estimate this for the system $G_{pub}$ of \texttt{Pesto}, we performed some computational experiments with MAGMA software \cite{magma} for different values of the parameters $n,m,t,s,q$.
For each set of parameters, we randomly generated $50$ systems and report the mean value of the solving degree.
The results are summarized in Table \ref{tab:degrees2}.

\begin{table}[h]
    \centering
\begin{tabular}{|c|c|c|c||c|c|c|}
    \hline
    \multirow{2}{1em}{$n$} &\multirow{2}{1em}{$m$}&\multirow{2}{1em}{$t$}&\multirow{2}{1em}{$s$}&\multicolumn{3}{|c|}{$ \sd(G_{pub})$} \\ \cline{5-7}
      &  &  &  &    $q=2^6$ & $q=2^8$ & $q=3761$\\
    \hline\hline
      7 & 5 & 2 & 2 & 6.00 &  6.00&  6.00 \\
    \hline
     7 & 6 & 2 & 1 & 6.00 &  6.00 &  6.00\\
    \hline
     8 & 5 & 3 & 3 & 7.00  & 7.00 & 7.00 \\
    \hline
     8 & 5 & 3 & 2 & 6.94  & 6.98  & 7.00 \\
    \hline
     8 & 6 & 3 & 2 & 15.02 &  15.02 & 15.00 \\
    \hline
     9 & 6 & 3 & 3 & 7.94 & 8.00& 8.00 \\
    \hline
     9 & 7 & 3 & 2 & 23.04 &  23.02 &23.00\\
    \hline
     10 & 6 & 3 & 3 & 8.00& 8.00& 8.00\\
    \hline
     10 & 7 & 3 & 2 & 7.00& 7.00& 7.00\\
    \hline
     10 & 8 & 3 & 1 & 13.00  & 13.00& 13.02 \\
    \hline
     10 & 9 & 2 & 1 & 9.00 & 9.00 & 9.00\\
    \hline
\end{tabular} 
    \caption{Mean values of the solving degree for 50 randomly generated systems.}
    \label{tab:degrees2}
\end{table}

\begin{remark}
We also conducted similar experiments for smaller field sizes, specifically for $q \leq 5$. As expected, in this scenario, it is advantageous to add the field equations $x_i^q - x_i = 0$ and $y_i^q - y_i = 0$ to the system $G_{pub}$ when computing the Gr\"obner basis. In several cases, this approach results in a lower solving degree compared to those reported in Table~\ref{tab:degrees2}.
However, this strategy becomes inefficient for larger values of $q$. Therefore, we recommend implementing the scheme over larger fields, specifically choosing $q \geq 2^6$. Moreover, since our goal is to analyse how the solving degree of our scheme behaves for different values of $q$ (including large $q$), we believe that studying the system without the field equations will offer better insight into how the solving degree scales.
\end{remark}

\subsection{Parameter proposals}\label{sec:parameters}
Taking into account all potential attacks explored so far, as well as the key size analysis in~\S\ref{sec:computationalrmks}, we propose specific parameter choices for the signature scheme \texttt{Pesto} from Definition~\ref{def:Pesto}. Based on our analysis, we set $t \approx \frac{n}{3}$ and $s = n - m = 2$. Under these assumptions, the previously outlined structural attacks do not appear to be applicable. In this scenario, the most effective attack remains the direct algebraic attack using Gr\"obner bases.
To mitigate this, we choose a large field size ($q =2^6$) to make adding field equations to the system computationally infeasible. As shown in Table~\ref{tab:degrees2}, the solving degree of the public system appears to grow rapidly and linearly with $n$, likely due to the presence of degree 4 equations in $G_{pub}$. Based on this observation, we conjecture possible estimates for the solving degree with larger parameters and use them to assess the complexity of a Gr\"obner basis attack, applying \eqref{eqGBcomplexity} with $\omega = 2.3$.

We then select the remaining parameters $(n, m, t)$ to align with NIST security levels for post-quantum signature schemes and compute the corresponding public and secret key sizes using the formulas from Remark~\ref{rmk:A1 reduced}. For comparison, we evaluate these key sizes against those of the UOV signature scheme \cite{UOVScheme} (without compression), which has been submitted to NIST’s Call for Additional Digital Signature Schemes in the PQC Standardization Process. The results are summarized in Table~\ref{tab:params}.

\begin{table}[h]
    \centering
\resizebox{\textwidth}{!}{ \begin{tabular}{|c||c|c|c|c|c|c|}
    \hline
       \makecell{Security \\level } & \makecell{Parameters \\$(\mathbb{F}_q ,n,m,t)$} &  \makecell{Estimated \\ $\sd(G_{pub})$} &  \makecell{\texttt{Pesto} PK \\ size (MB)} & \makecell{\texttt{Pesto} SK \\ size (MB)} &\makecell{UOV public \\ key  size (MB)} &\makecell{UOV secret \\ key  size (MB)} \\
        \hline\hline
     NIST I  &$(\mathbb{F}_{2^6}, 27,25,10)$ & $\geq 33$ & 357  & 5.5& 278.5 &237.9\\
     \hline
     NIST III &$(\mathbb{F}_{2^6}, 40,38,14)$ & $\geq48$& 2453 & 16.5&1225.4 &1044.3\\
     \hline
     NIST V  & $(\mathbb{F}_{2^6},57,55,20)$ & $\geq 59$ &13725 & 43.8 &2869.4 &2436.7 \\
     \hline
    \end{tabular}}
    \caption{Parameter proposals for the signature scheme \texttt{Pesto} of Definition~\ref{def:Pesto}. }
    \label{tab:params}
\end{table}

\begin{remark}
We would like to emphasize that the parameters in Table~\ref{tab:params} are merely indicative of how the key sizes of \texttt{Pesto} might grow. In particular, they highlight that there appears to be a clear advantage in the secret key size, while the public key size increases rapidly, which could make the scheme more suitable for lower security levels. However, at this stage of the proposal, it remains uncertain whether our assumptions about how the solving degree scales with the parameters are correct. We therefore encourage further cryptanalysis and additional study of the scheme to better assess its security.
\end{remark}

\section{Challenges and open problems}
In this work, we propose applying a CCZ transformation in the construction of a multivariate scheme, as an alternative to the usual affine transformation. This offers the advantage of concealing linear relationships between the input and output that would typically arise with an affine transformation alone.
To better evaluate the stability and feasibility of this approach, we have identified a series of open problems and questions relevant to future cryptanalysis.

\begin{enumerate}
    \item How does the solving degree grow with respect to the input parameters? Due to the limited set of experiments we were able to conduct, predicting the behavior of the solving degree of the public polynomial system remains challenging. Based on our experiments, it appears that the solving degree scales linearly with the number of variables $n$. However, further investigation and, hopefully, a mathematical proof are desirable.

    \item Is it possible to combine the CCZ construction with other known multivariate primitives? More specifically, can the map $U$ be chosen using HFE polynomials? How does the resulting scheme compare to the \texttt{Pesto} scheme of Definition~\ref{def:Pesto}?

    \item None of the attacks presented in Section~\ref{sec:securityanalysis} appear to fully compromise the scheme. However, we believe that further cryptanalysis is necessary to better assess the security of the scheme. In particular, can classical attacks on Oil and Vinegar schemes be adapted to the \texttt{Pesto} scheme?

    \item The public key consists of degree 4 polynomials and, as a result, appears to be quite large. At present, this is the main practical disadvantage of the scheme. Could techniques used for key size reduction in OV schemes, such as the QR technique from \cite{FIKT21}, be adapted to reduce the key size?
    
\end{enumerate}

As the previous problems highlight, we see significant potential for further exploration in this direction, which we hope could build a fruitful bridge between the areas of cryptographic Boolean functions and multivariate cryptography. Since both fields involve functions and polynomials defined over finite fields, we believe that many techniques used in one area could be studied and applied in the other.

\bibliography{biblio}

\end{document}